\documentclass[11pt,a4paper]{article}
\usepackage{fullpage,amssymb,amsmath,amsfonts,amsthm,paralist,hyperref,graphicx,color}
\usepackage[normalem]{ulem}

\newtheorem{lemma}{Lemma}
\newtheorem{theorem}{Theorem}
\newtheorem{corollary}{Corollary}
\newtheorem{definition}{Definition}
\newtheorem{conjecture}{Conjecture}

\newcommand{\XSays}[3]{{\color{#2}
      {$\rule[-0.12cm]{0.2in}{0.5cm}$\fbox{\tt
            #1:} }%
      \itshape #3
      \marginpar{\color{#2}\tt #1}%
      \def\comment{#3}\def\empty{}\ifx\comment\empty\else
      {$\rule[0.1cm]{0.3in}{0.1cm}$\fbox{\tt
            end}$\rule[0.1cm]{0.3in}{0.1cm}$} \fi
   }%
}

\title{No acute tetrahedron is an 8-reptile}
\author{Herman Haverkort\thanks{Dept. of Mathematics and Computer Science, Eindhoven University of Technology, the Netherlands,\hfill\break cs.herman@haverkort.net}}
\date{}

\begin{document}
\maketitle

\begin{abstract}
An $r$-gentiling is a dissection of a shape into $r \geq 2$ parts which are all similar to the original shape. An $r$-reptiling is an $r$-gentiling of which all parts are mutually congruent. The complete characterization of all reptile tetrahedra has been a long-standing open problem. This note concerns acute tetrahedra in particular. We find that no acute tetrahedron is an $r$-gentile or $r$-reptile for any $r < 10$. The proof is based on showing that no acute spherical diangle can be dissected into less than ten acute spherical triangles.
\end{abstract}

\section*{Introduction}

Let $T$ be a closed set of points in Euclidean space with a non-empty interior. We call $T$ an \emph{$r$-gentile} if $T$ admits an \emph{$r$-gentiling}, that is, a subdivision of $T$ into $r \geq 2$ sets (\emph{tiles}) $T_1,...,T_r$, such that each of the sets $T_1,...,T_r$ is similar to $T$. In other words, $T$ is an $r$-gentile if we can tile it with $r$ smaller copies of itself. This generalizes the concept of \emph{reptiles}, coined by Golomb~\cite{golomb}: a set $T$ is an \emph{$r$-reptile} if $T$ admits an \emph{$r$-reptiling}, that is, a subdivision of $T$ into $r \geq 2$ sets $T_1,...,T_r$, such that each of the sets $T_1,...,T_r$ is similar to $T$ \emph{and all sets $T_1,...,T_r$ are mutually congruent under translation, rotation and/or reflection.} In other words, $T$ is an $r$-reptile if we can tile it with $r$ equally large, possibly reflected, smaller copies of itself. Interest in reptile tetrahedra (or triangles, for that matter) exists, among other reasons, because of their application in meshes for scientific computing~\cite{bader,Liu}. In this realm techniques such as reptile-based stack-and-stream are well-developed in two dimensions, but three-dimensional space poses great challenges~\cite{bader}.

It is known what triangles are $r$-reptiles~\cite{snover} and $r$-gentiles~\cite{freese,kaiser} for what $r$. However, for tetrahedra the situation is much less clear; in fact the identification of reptile and gentile tetrahedra and, even more general, of tetrahedra that tile space, has been a long-standing open problem~\cite{senechal}. 

The regular tetrahedron does \emph{not} tile space, as its dihedral angles are $\arccos(1/3)$, which is larger than $2\pi/6$ but slightly smaller than $2\pi/5$, so that no number of regular tetrahedra can fill the space around a common edge. Goldberg described all \emph{known} tetrahedra that do tile space~\cite{goldberg}. Delgado Friedrichs and Huson characterize all tetrahedra that produce \emph{tile-transitive} tilings~\cite{delgado}, but to the best of my knowledge, without the restriction to tile-transitive tilings the problem of identifying all space-filing tetrahedra is still open.

The reptile tetrahedra must be a subset of the tetrahedra that tile space. Matou\v sek and Safernov\'a argued that $r$-reptilings with tetrahedra exist if and only if $r$ is a cube number~\cite{safernova}. In particular, it is known that all so-called \emph{Hill tetrahedra} (attributed to Hill~\cite{Hill} by Hertel~\cite{Hertel} and Matou\v sek and Safernov\'a~\cite{safernova}) are 8-reptiles. It has been conjectured that the Hill tetrahedra are the only reptile tetrahedra~\cite{Hertel}, but this conjecture is false: Sommerville found two non-Hill tetrahedra that tile three-dimensional space~\cite{Sommerville} and which were recognized as 8-reptiles by Liu and Joe~\cite{Liu}. To the best of my knowledge, the Hill tetrahedra and the two non-Hill tetrahedra from Liu and Joe are the only tetrahedra known to be reptiles, but there might be others. This paper provides a small contribution to the answer to the question: exactly what tetrahedra are reptiles?

In mesh construction applications one typically needs to enforce certain quality constraints on the mesh elements. This has motivated studies into \emph{acute tetrahedra}~\cite{ungor}:

\begin{definition}
A tetrahedron is \emph{acute} if each pair of its facets has a dihedral angle strictly less than $\pi/2$.
\end{definition}

All facets of an acute tetrahedron are acute triangles themselves (Eppstein et al.~\cite{ungor}, Lemma~2). The Hill tetrahedra, as well as the two non-Hill tetrahedra from Liu and Joe, all have right dihedral angles\footnote{One of the non-Hill tetrahedra can be given (modulo similarity transformations) by $A = (-1,0,0)$, $B = (0,1,0)$, $C = (1,0,0)$, $D = (0,0,\sqrt{1/2})$; the second type of non-Hill tetrahedron is obtained by cutting the first type along the $yz$-plane. Both types have right dihedral angles along the $x$-axis. Any Hill tetrahedron can be described as the convex hull of four vertices $A = 0$, $B = v_1$, $C = v_1+v_2$ and $D = v_1+v_2+v_3$, such that the vectors $v_1$, $v_2$ and $v_3$ have the same length and such that the angle between each pair of these vectors is the same, say $\alpha$~\cite{safernova}. For ease of notation, assume that the tetrahedron is scaled, rotated and reflected such that $v_1$, $v_2$ and $v_3$ have length $\sqrt 2$, the vertex $C = v_1 + v_2$ lies on the positive $x$-axis, and the vertex $B = v_1$ lies in the first quadrant of the $xy$-plane.
We use $t$ to denote $\cos\alpha$. Note that we have $t < 1$, otherwise we would have $\alpha = 0$, all vertices would lie on a single line, and they would not be the vertices of a tetrahedron. The condition on the angles of the vectors can now be written as $v_1 \cdot v_2 = v_1 \cdot v_3 = v_2 \cdot v_3 = 2t$.
Thus we must have $v_1 = (a, b, 0)$ and $v_2 = (a, -b, 0)$ with $a = \sqrt{1 + t}$ and $b = \sqrt{1 - t} > 0$, so that indeed, $||v_1|| = ||v_2|| = \sqrt{a^2 + b^2} = \sqrt 2$ and $v_1 \cdot v_2 = a^2 - b^2 = 2t$. The vector $v_3 = (x, y, z)$ must now satisfy $v_1 \cdot v_3 = v_2 \cdot v_3 \Leftrightarrow ax + by = ax - by$, which, given $b \neq 0$, solves to $y = 0$, and we get $D = v_1+v_2+v_3 = (2a+x, 0, z)$. Thus, the face $ABC$ lies in the $xy$-plane and the face $ACD$ lies in the $xz$-plane, and these faces meet at a right dihedral angle along the $x$-axis.}.
Thus, no acute reptile tetrahedra are known.

\section*{Results}

In this note we will prove the following statement, which may serve as evidence that acute reptile tetrahedra are probably hard to find, if they exist at all:

\begin{theorem}\label{thm:main}
Let $T$ be an acute tetrahedron subdivided into $r \geq 2$ acute tetrahedra $T_1,...,T_r$. If the diameter (longest edge) of each tetrahedron $T_i$ is smaller than the diameter (longest edge) of $T$, then $r \geq 10$.
\end{theorem}

In particular we get:
\begin{corollary}\label{cor:gentile}
No acute tetrahedron is an $r$-gentile for any $r < 10$.
\end{corollary}

With the result from Matou\v sek and Safernov\'a that $r$-reptile tetrahedra can only exist when $r$
is a cube number~\cite{safernova}, we get:

\begin{corollary}\label{cor:reptile}
No acute tetrahedron is an $r$-reptile for any $r < 27$.
\end{corollary}

\section*{The proof}

Note that if a tetrahedron $T$ is subdivided into tetrahedra $T_1,...,T_r$ with smaller diameter than $T$, then at least one tetrahedron $T_i$, for some $i \in \{1,...,r\}$, must have a vertex $v$ on the longest edge of $T$. For the proof of Theorem~\ref{thm:main} we analyse ${\cal S}_v$, the subdivision of an infinitesimal sphere around $v$ that is induced by the facets of $T$ and $T_1,...,T_r$. In such a subdivision, we find:\begin{itemize}
\item faces: each face is either a spherical triangle, corresponding to a tetrahedron $T_i$ of which $v$ is a vertex, or a spherical diangle (also called lune), corresponding to a tetrahedron that has $v$ on the interior of an edge;
\item edges: the edges of ${\cal S}_v$ are segments of great circles and correspond to facets of $T_1,...,T_r$ that contain $v$; the angle between two adjacent edges on a face of ${\cal S}_v$ corresponds to the dihedral angle of the corresponding facets of a tetrahedron $T_i$.
\item vertices: each vertex of ${\cal S}_v$ corresponds to an edge of a tetrahedron $T_i$ that contains $v$.
\end{itemize}
Thus, ${\cal S}_v$ consists of a spherical diangle $D$ corresponding to $T$, subdivided into a number of spherical triangles, and possibly some spherical diangles, that correspond to the tetrahedra from $T_1,...,T_r$ that touch $v$. Below we will see that ${\cal S}_v$ must contain at least ten faces (not counting the outer face, that is, the complement of $D$), which proves Theorem~\ref{thm:main}.

In what follows, when we talk about diangles and triangles, we will mean \emph{acute}, \emph{spherical} diangles and \emph{acute}, \emph{spherical} triangles on a sphere with radius~1. Note that the faces are diangles or triangles in the geometric sense, but they may have more than two or three vertices on their boundary. More precisely, a diangle or triangle has, respectively, exactly two or three vertices, called \emph{corners}, where its boundary has an acute angle, and possibly a number of other vertices where its boundary has a straight angle. A chain of edges of a diangle or triangle from one corner to the next is called a \emph{side}. Note that ${\cal S}_v$ contains at least one triangle, since $v$ is a vertex of at least one tetrahedron $T_i$. Therefore, in what follows we consider a subdivision ${\cal S}$ of a diangle $D$ into a number of diangles and triangles, among which at least one triangle. We call such subdivisions \emph{valid}. Henceforth, we will assume that ${\cal S}$ has the smallest number of faces out of all possible valid subdivisions of all possible diangles $D$. Our goal is now to prove that ${\cal S}$ contains at least ten faces.

\begin{lemma}\label{lem:alltriangles}
Each face of ${\cal S}$ is a triangle.
\end{lemma}
\begin{proof}
If ${\cal S}$ would have any diangular face $F$, it must have the same corners as $D$, because the corners of any diangle must be an antipodal pair and there is only one antipodal pair within $D$. The removal of $F$ would separate ${\cal S}$ into at most two diangular components with the same corners as $D$. By construction, at least one of these components contains a triangle. That component would then constitute a valid subdivision that has fewer faces than ${\cal S}$, contradicting our choice of ${\cal S}$.
\end{proof}

In ${\cal S}$, we distinguish \emph{boundary vertices} (vertices on the boundary of $D$) and \emph{interior vertices} (vertices in the interior of $D$). Among the boundary vertices, we distinguish \emph{poles} (the corners of $D$) and \emph{side vertices} (the remaining boundary vertices). Among the interior vertices we distinguish \emph{full vertices} and \emph{hanging vertices}: a vertex $v$ is a full vertex if it is a corner of each face incident on $v$; a vertex $v$ is a hanging vertex if it is a non-corner vertex of one of the faces incident on $v$.

We will now derive a few properties of ${\cal S}$ from the acuteness of its angles.

\begin{lemma}\label{lem:shortedges}
Each side of each face of ${\cal S}$ has length strictly less than $\pi/2$.
\end{lemma}
\begin{proof}
Consider any face $F$ of ${\cal S}$. Let $a$ be the length of a particular side of $F$, let $\alpha$ be the angle in the opposite corner of $F$, and let $\beta$ and $\gamma$ be the angles in the other two corners of~$F$. Since $F$ is acute, the sines and cosines of $\alpha, \beta$ and $\gamma$ are all positive. By the supplementary cosine rule (\cite{todhunter}, Art.~47) we have $\cos \alpha = -\cos\beta\cos\gamma + \sin\beta\sin\gamma\cos a$, so $\cos a = (\cos\alpha + \cos\beta\cos\gamma)/(\sin\beta\sin\gamma) > 0$. It follows that $a < \pi/2$.
\end{proof}

\begin{lemma}\label{lem:fourboundary}
There are at least four side vertices: two on each side of $D$.
\end{lemma}
\begin{proof}
For the sake of contradiction, suppose one side of $D$ contains only one side vertex. Then this side would consist of two edges, at least one of which has length at least $\pi/2$, contradicting Lemma~\ref{lem:shortedges}.
\end{proof}

\begin{lemma}\label{lem:mindegree}
Each pole is incident on at least two edges.\\
Each hanging vertex and each side vertex is incident on at least four edges.\\
Each full vertex is incident on at least five edges.
\end{lemma}
\begin{proof}
The poles are incident on at least two edges by definition. If a side vertex or a hanging vertex would be incident on only three edges, then two of these edges make a straight angle on one side, while the third edge divides the straight angle on the other side. Thus, at least one of the angles that results from this division would be non-acute. If a full vertex would be incident on at most four edges, then at least two of those edges must make an angle of at least $2\pi/4 = \pi/2$ on their common face, again contradicting the assumption that all faces are acute.
\end{proof}

Now we can combine these properties with Euler's formula and find:

\begin{lemma}\label{lem:minfaces}
The number of faces equals $2f + h + s$, where $f \geq 2$ is the number of full vertices, $h$ is the number of hanging vertices, and $s \geq 4$ is the number of side vertices. More precisely, $f \geq 2 + p$, where $p$ is the number of edges with one end point at a pole and the other end point at an interior vertex.
\end{lemma}
\begin{proof}
Let $v = f + h + s + 2$ be the number of vertices, let $e$ be the number of edges and $r$ be the number of triangles of ${\cal S}$. By Lemma~\ref{lem:alltriangles} all faces are triangles, so by Euler's formula we have $v + r = f + h + s + 2 + r = e + 1$, hence $2e = 2f + 2h + 2s + 2 + 2r$. We say that a hanging vertex is \emph{owned} by the triangle of which it is a non-corner vertex; each hanging vertex is owned by exactly one triangle. The number of edges on the boundary of a triangle $F$ is three plus the number of hanging vertices owned by $F$. Thus, if we add up the numbers of edges of all triangles, we obtain a total of $3r + h$. This counts all edges double, except the $s + 2$ edges on the boundary of $D$, which are counted only once. Therefore we have $2e - s - 2 = 3r + h$. Hence we have $2e = 3r + h + s + 2 = 2f + 2h + 2s + 2 + 2r$, which solves to $r = 2f + h + s$.

Thus we have $2e = 3r + h + s + 2 = 6f + 4h + 4s + 2$. By Lemma~\ref{lem:mindegree} we also have $2e \geq 5f + 4h + 4s + 4 + p$, so $6f + 4h + 4s + 2 \geq 5f + 4h + 4s + 4 + p$, which solves to $f \geq 2 + p$.

The condition $s \geq 4$ is given by Lemma~\ref{lem:fourboundary}.
\end{proof}

\begin{lemma}\label{lem:thelemma}
The number of faces of ${\cal S}$ is at least ten.
\end{lemma}
\begin{proof}
Suppose, for the sake of contradiction, that ${\cal S}$ has at most nine faces. Then, by Lemma~\ref{lem:minfaces}, we cannot have $f \geq 3$, so $f = 2$, that is, there are exactly two full vertices, which we denote by $X$ and $Y$. Let $N_X$ and $N_Y$ be the vertices adjacent to $X$ and $Y$, respectively. Moreover, by Lemma~\ref{lem:minfaces}, we have $p = 0$ (so neither $N_X$ nor $N_Y$ contains any of the poles), $f + s + h \leq 7$ (there are at most seven vertices other than the poles), and in particular $h \leq 1$ (there is at most one interior vertex other than $X$ and $Y$).

By Lemma~\ref{lem:mindegree}, we have $|N_X| \geq 5$ and $|N_Y| \geq 5$. Since $|N_X \cup N_Y| \leq f + s + h \leq 7$, this implies $|N_X \cap N_Y| \geq 3$, that is, $X$ and $Y$ share at least three neighbours. Since there is at most one interior vertex other than $X$ and $Y$, at least two of the shared neighbours must be side vertices, let us call these $V$ and $W$. Now consider the cycle $XVYWX$. This cycle separates the diangle $D$ into three regions---see Figure~\ref{fig:hopeless}a. Only one of these regions, namely the interior of the cycle, has both $X$ and $Y$ on its boundary. Therefore, the third shared neighbour of $X$ and $Y$ must be an interior vertex $U$ inside the cycle $XVYWX$. Since there are no other interior vertices, $U$ can have edges only to $X$, $V$, $Y$, and $W$. By Lemma~\ref{lem:mindegree}, all four of these edges must exist.

\begin{figure}
\centering\hbox to\hsize{%
a)~\includegraphics[scale=0.8,page=2]{hopeless.pdf}\hfill
b)~\includegraphics[scale=0.8,page=1]{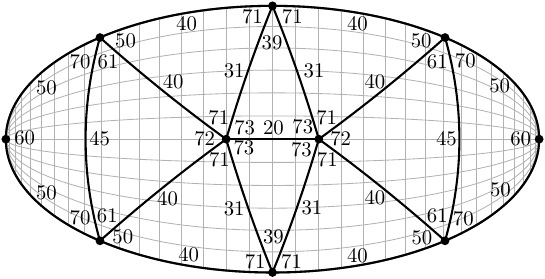}%
}
\caption{
a) If ${\cal S}$ has at most nine faces, two full vertices and two side vertices must form a quadrilateral with a hanging vertex inside, but we cannot get the angles around the hanging vertex right.\quad b) An example of a subdivision of an acute diangle into ten acute triangles, with angles and edge lengths given in degrees rounded to integers.}
\label{fig:hopeless}
\end{figure}

Now, since $U$ is a hanging vertex, one of the angles around $U$ must be a straight angle, say, without loss of generality, the angle $XUV$. But then the cycle $XUVX$ would constitute a geometric diangle, contradicting Lemma~\ref{lem:alltriangles}. Therefore, the assumption that ${\cal S}$ has at most nine faces must be wrong.
\end{proof}

This concludes the proof of Theorem~\ref{thm:main}.

Note that the crucial observation of the proof is that any dissection of an acute spherical diangle into acute spherical triangles requires at least ten triangles. This bound is, in general, tight: Figure~\ref{fig:hopeless}b shows a dissection of an acute spherical diangle into ten acute spherical triangles.

\section*{Bold conjectures}

Could we exploit the observations in this paper to improve Corollary~\ref{cor:reptile} further? A first step could be the following. If $T$ is an acute tetrahedron with diameter $d$, and we can identify three segments $x$, $y$ and $z$ of length $d/3$ on the edges of $T$, then the arguments presented in this paper tell us that any 27-reptiling of $T$ must contain at least ten tiles that intersect $x$, ten tiles that intersect $y$, and ten tiles that intersect $z$. If additionally, one can ensure that $x$, $y$ and $z$ lie at distance more than $d/3$ from each other, these three sets of ten tiles each must be mutually disjoint, so there must be at least 30 tiles in total. This would contradict the existence of a 27-reptiling and thus improve Corollary~\ref{cor:reptile} to: no acute tetrahedron is an $r$-reptile for any $r < 64$.

However, given that no acute tetrahedron can be an $r$-reptile or $r$-gentile for small values of $r$ (and given, in general, that the past hundred years did not turn up any reptile tetrahedra without right dihedral angles), we may rather restate the obvious:

\begin{conjecture}
There are no reptile acute tetrahedra.
\end{conjecture}

A stronger conjecture would be:

\begin{conjecture}
There are no gentile acute tetrahedra.
\end{conjecture}

We conclude with an even bolder conjecture:

\begin{conjecture}
There are no gentile tetrahedra that do not have a dihedral angle of exactly $\pi/2$.
\end{conjecture}

\paragraph{Acknowledgements}
The author thanks the anonymous reviewer who improved Lemma~\ref{lem:minfaces} and gave me a hint on how to use this to improve Lemma~\ref{lem:thelemma}. This raised the bound in Lemma~\ref{lem:thelemma} and thus, many other bounds in this paper, from nine to ten.

\bibliographystyle{amsplain}

\begin{thebibliography}{10}

\bibitem{bader}
M.~Bader and C.~Zenger. Efficient storage and processing of adaptive
  triangular grids using {S}ierpinski curves. \emph{6th International Conference on
  Computational Science (ICCS)}, Lecture Notes in Computer Science 3991:673--680 (2006).

\bibitem{delgado}
O.~Delgado Friedrichs and D.~H.~Huson. Tiling Space by Platonic Solids, I. \emph{Discrete Comput. Geom} \textbf{21}:299–-315 (1999).

\bibitem{ungor}
D.~Eppstein, J.~M. Sullivan, and A.~{\"U}ng{\"o}r. Tiling space and slabs
  with acute tetrahedra. \emph{Computational Geometry} \textbf{27}(3):237--255 (2004).

\bibitem{freese}
R.~W. Freese, A.~K. Miller, and Z.~Usiskin. Can every triangle be divided
  into $n$ triangles similar to it? \emph{Am. Math. Monthly} \textbf{77}:867--869 (1990).

\bibitem{goldberg}
M. Goldberg. Three infinite families of tetrahedral space-fillers. \emph{J. Combin. Theory} 16:348--354 (1974).

\bibitem{golomb}
S.~W. Golomb. \emph{Replicating figures in the plane}. \emph{The Mathematical Gazette} \textbf{48}(366):403--412 (1964).

\bibitem{Hertel}
E.~Hertel. Self-similar simplices. \emph{Beitr{\"a}ge zur Algebra und Geometrie} \textbf{41}(2):589--595 (2000).

\bibitem{Hill}
M.~J.~M. Hill. Determination of the volumes of certain species of tetrahedra. \emph{Proc. London Mathematical Society} \textbf{27}:39--53 (1896).

\bibitem{kaiser}
H.~Kaiser. \emph{Selbst{\"a}hnliche {D}reieckszerlegungen}. Tech. report, Friedrich-Schiller-Universit{\"a}t Jena, 1990.

\bibitem{Liu}
A.~Liu and B.~Joe, On the shape of tetrahedra from bisection. \emph{Mathematics of Computation} \textbf{63}(207):141--154 (1994).

\bibitem{safernova}
J.~Matou{\v s}ek and Z.~Safernov{\'a}. On the nonexistence of {$k$}-reptile tetrahedra. \emph{Discrete \& Computational Geometry} \textbf{46}(3):599--609 (2011).

\bibitem{senechal}
M.~Senechal. Which tetrahedra fill space? \emph{Math. Magazine} 54:227–-243 (1981).

\bibitem{snover}
S.~L. Snover, C.~Waiveris, and J.~K. Williams. Rep-tiling for triangles. \emph{Discrete Mathematics} \textbf{91}:193--200 (1991).

\bibitem{Sommerville}
D.~M.~Y. Sommerville. Division of Space by Congruent Triangles and Tetrahedra. \emph{Proc. Royal Soc. of Edinburgh} \textbf{43}:85--116 (1924).

\bibitem{todhunter}
I.~Todhunter. \emph{Spherical trigonometry}, 5th ed., MacMillan and Co., 1886, \url{http://www.gutenberg.org/ebooks/19770}.

\end{thebibliography}
\providecommand{\bysame}{\leavevmode\hbox to3em{\hrulefill}\thinspace}
\providecommand{\MR}{\relax\ifhmode\unskip\space\fi MR }
\providecommand{\MRhref}[2]{%
  \href{http://www.ams.org/mathscinet-getitem?mr=#1}{#2}
}
\providecommand{\href}[2]{#2}

\end{document}